 \documentclass[final,5p,times,twocolumn,authoryear]{elsarticle}

\usepackage{amssymb}

\usepackage{amsmath,amssymb,amsfonts}
\usepackage[justification=centering]{caption}
\usepackage{subfigure}
\usepackage{verbatim}
\usepackage{hyperref}
\hypersetup{
	colorlinks=true,
	linkcolor=blue,
	filecolor=blue,      
	urlcolor=blue,
	citecolor=cyan,
}

\newdefinition{assumption}{ {Assumption}}
\newtheorem{theorem}{ {Theorem}}
\newtheorem{lemma}{ {Lemma}}
\newdefinition{definition}{ {Definition}}
\newdefinition{remark}{ {Remark}}
\newdefinition{example}{ {Example}}
\newproof{proof}{Proof}

\begin{document}

\begin{frontmatter}



\title{Exponential convergence of distributed optimization for heterogeneous linear multi-agent systems \tnoteref{t1}}
\tnotetext[t1]{This research was supported by the Shanghai Municipal Commission of Science and Technology No. 19511132100, 19511132101, the National Natural Science Foundation of China under Grant 62003243, and National Key R\&D Program of China, No. 2018YFE0105000, 2018YFB1305304.}

\author[1,2]{Li Li}
\ead{lili@tongji.edu.cn}
\author[1]{Yang Yu}
\ead{1910639@tongji.edu.cn}
\author[1,3]{Xiuxian Li\corref{cor1}}
\ead{xli@tongji.edu.cn}
\author[4]{Lihua Xie
	}
\ead{elhxie@ntu.edu.sg}

\cortext[cor1]{Corresponding author.}
\address[1]{Department of Control Science and Engineering, College of Electronics and Information Engineering, Tongji University, Shanghai 201804, China}
\address[2]{Shanghai Institute of Intelligent Science and Technology, Tongji University, Shanghai 201804, China}
\address[3]{Institute for Advanced Study, Tongji University, Shanghai, 200092, China}
\address[4]{School of Electrical and Electronic Engineering, Nanyang Technological University, Singapore 639798}


\begin{abstract}
In this work we study a distributed optimal output consensus problem for heterogeneous linear multi-agent systems where the agents aim to reach consensus with the purpose of minimizing the sum of private convex costs. Based on output feedback, a fully distributed control law is proposed by using the proportional-integral (PI) control technique. For strongly convex cost functions with Lipschitz gradients, the designed controller can achieve convergence exponentially in an undirected and connected network. 
Furthermore, to remove the requirement of continuous communications, the proposed control law is then extended to periodic and event-triggered communication schemes, which also achieve convergence exponentially. Two simulation examples are given to verify the proposed control algorithms.
\end{abstract}




\end{frontmatter}


\section{Introduction}\label{int}
In recent few decades, distributed optimization has been attracting more and more research interests because of its wide applications in multi-agent systems, smart grids, machine learning and so on. Specifically, the purpose of each node in a network is to minimize the sum of private costs under constraints only by exchanging local information with neighbors. Various distributed algorithms have been proposed in this field \citep{nedic2009,yang2019,zhang2017,tang2019,zhao2017,li20191,li20192}.

In practical physical systems such as AGVs and UAVs, the implementation of distributed strategies must consider the dynamics of each agent. Therefore, in recent years, interest has been attracted increasingly in distributed optimization combined with physical systems. This problem requires each of a group of continuous-time physical systems to achieve the best performance. Generally, there are two ways to solve such problems. The first one is based on a “separative” way: treating it as a standalone distributed optimization problem for cost functions and simultaneously tracking the optimized variables for complex systems. In \citet{zhang2017} where the system dynamics are Euler-Lagrange systems, two distributed algorithms are developed for the case without parametric uncertainties and the case with parametric uncertainties respectively. For high-order multi-agent systems, \citet{tang2019} firstly makes use of a virtual first-order optimizer to generate an optimal signal and then uses an underlying controller to make the system track this signal. This method is relatively simple due to the mature research of two dimensions, but it is an open-loop structure. The second one focuses on designing an integrated control law. Reference \citet{zhao2017} studies the optimal consensus problem of linear systems, but it requires the local objective function to be of a certain special structure, which only works in limited situations. For heterogeneous/homogeneous linear multi-agent systems, asymptotically stable and fully distributed controllers are designed in \citet{li2020,zhang2020}. This kind of structural design is more complicated, but the integrated design makes it possible to optimize the performance the converged state of the system.

In continuous-time distributed optimization algorithms, it is necessary for agents to exchange information continuously, which is unrealistic in actual physical systems. In order to avoid continuous communication and reduce communication overhead, both periodic and event-triggered mechanisms \citep{ding2018,ge2020} are introduced for first-order integrator\citep{kia2015,wu2020}, second-order integrator\citep{tran2018}, linear multi-agent systems\citep{li2020,lina2017,lxx2019} and so on. Each agent communicates with its neighbors periodically or only after reaching certain trigger conditions. For the event-triggered communication mechanism, one of the key points is to avoid Zeno behavior, meaning that an infinite number of events occur in a finite time. One of the methods is to design a lower bound of the communication interval as in \citet{kia2015}, but it can only ensure that the algorithm converges to a neighborhood of the optimal solution. Reference \citet{yu2020} studies the Zeno behavior of the first-order multi-agent systems and provides sufficient conditions for its existence in a finite time consensus.

Aiming at the optimal output consensus problem of heterogeneous linear multi-agent systems, this paper  designs a proportional-integral (PI) controller to solve the problem, in which the proportional term drives all the agents to the consensus space and the integral term eliminates errors\citep{qiu2019}. The main contributions of this paper are as follows.
\begin{itemize}	
	\item[1)] Through the feedback combination of own state and neighbor output information, a PI based control law is designed, which is shown to have an exponential convergence rate. In comparison, the most related work \citep{li2020} only gives the result of asymptotic convergence, and is based on a stronger assumption than that in this paper.
	
	\item[2)] To reduce the communication overhead, this paper further introduces periodic communication and event-triggered communication mechanisms for the above-proposed algorithm, which are both proven to guarantee exponential convergence. Besides the established exponential rate here, compared with the gradually decreasing communication interval in \citet{li2020}, the proposed algorithm can clearly give a lower bound of the communication interval, thus excluding the Zeno behavior. 
\end{itemize}

The overall structure of this paper is as follows. Preliminaries are given in Section \ref{pre}. In Section \ref{pro}, the heterogeneous multi-agent systems under investigation are described mathematically, the optimal output consensus problem is defined and some useful lemmas are given. Following that, three control laws with continuous, periodic, and event-triggered communication are proposed respectively and their exponential convergence is established in Section \ref{mai}. Then two simulation examples are provided to verify the effectiveness of the algorithms in Section \ref{sim}. Finally, conclusions and future works are discussed in Section \ref{res}.

\section{Preliminaries}\label{pre}

\subsection{Notations}
Let $\mathbb{R}$, $\mathbb{R}^n$, $\mathbb{R}^{m\times n}$ be the sets of real numbers, real vectors of dimension $n$ and real matrices of dimension $m\times n$, respectively. $I_n$ denotes the $n$-dimensional identity matrix. $\boldsymbol1_n$ and $\boldsymbol0_n$ denote $n$-dimensional all-one and all-zero column vectors, respectively. For a matrix $A\in\mathbb{R}^{m \times n}$, $A^\top $ is its transpose, and  $diag(A_1,\dots,A_n)=blkdiag(A1; \dots; An)$ denotes a block diagonal matrix with diagonal blocks of $A_1$, $\dots$, $A_n$ . $col(x_1,\dots,x_n) = (x_1^\top ,\dots,x_n^\top )^\top $ is a column vector by stacking vectors $x_1,\dots,x_n$. $\|A\|$ and $\|x\|$ are the induced 2-norm of matrix $A$ and the Euclidean norm of vector $x$ respectively. $A\otimes B$ represents the Kronecker product of matrices A and B. 

\subsection{Graph Theory}
A communication network of $N$ agents is modeled by an undirected graph $\mathcal{G} = (\mathcal{V}, \mathcal{E}, \mathcal{A})$, where $\mathcal{V} = \{v_1, \dots, v_N\}$ is a node set, $\mathcal{E} \in \mathcal{V} \times \mathcal{V}$ is an edge set and $\mathcal{A}= [a_{ij}]\in \mathbb{R}^{N \times N}$ is the adjacency matrix. If information exchange can occur between $v_i$ and $v_j$, then $(v_i,v_j) \in \mathcal{E}$. If there exists a path from any node to any other node in $\mathcal{V}$, then $\mathcal{G}$ is called connected, otherwise disconnected. The out-degree of node $v_i$ is denoted by $d_i^{out} = \sum_{j=1}^{N}a_{ji}$. Denote $L=D^{out}-A$ as the Laplacian matrix of $\mathcal{G}$, where $D^{out}=diag(d_1^{out}, \dots, d_N^{out})$ is the out-degree matrix of $\mathcal{G}$.

\begin{lemma}[\citealp{tu}]\label{lemma1}
	If $\mathcal{G}$ is undirected and connected, all eigenvalues  of $L$ are real and except for a single eigenvalue $0$, the rest are positive numbers, denoted as
	$$0=\lambda_1<\lambda_2\le \dots \le \lambda_N.$$
\end{lemma}

\begin{lemma}[\citealp{li2020}]\label{lemma2}
	If $\mathcal{G}$ is undirected and connected, there exists a positive definite matrix $\Gamma \in \mathbb{R}^{N\times N}$ such that $\Gamma L = L\Gamma=\Pi$, where $\Pi=I_N-\frac{1}{N}\boldsymbol1_N\boldsymbol1_N^\top $.
	Moreover, the eigenvalues of $\Gamma$ are $\lambda_\Gamma, \frac{1}{\lambda_2}, \dots, \frac{1}{\lambda_N}$, where $\lambda_\Gamma>0$ can be any positive constant and $\lambda_2, \dots, \lambda_N$ are defined in Lemma $1$.
\end{lemma}

\section{Problem Formulation}\label{pro}

Consider a multi-agent system with $N$ heterogeneous agents, and the $i$th agent has the linear dynamics:
\begin{align}
	\begin{split}
		{\dot x}_i&=A_i x_i + B_i u_i,\\
		y_i&=C_i x_i,\label{eqsys}
	\end{split}
\end{align}
where $x_i\in\mathbb{R}^{n_i}$, $u_i\in\mathbb{R}^{p_i}$ and $y_i\in\mathbb{R}^{q}$ are the state, input and output variables respectively. $A_i\in\mathbb{R}^{n_i \times n_i}$, $B_i\in\mathbb{R}^{n_i \times p_i}$ and $C_i\in\mathbb{R}^{q \times n_i}$ are the state, input and output matrices which are constant.

The objective of this paper is to design a controller $u_i(t)$ for each agent by using only local interaction and information such that all agents cooperatively reach the optimal outputs that solve the following convex optimization problem:
\begin{align}
	\min\sum_{i=1}^N f_i(y_i),\ s.t.\ y_1=\dots =y_N,\label{eqfun}
\end{align}
where $f_i:\mathbb{R}^{q}\to\mathbb{R}$ is the local cost function which is only known to the $i$th agent.

\begin{assumption} \label{asp1}
	The communication network $\mathcal{G}$ is undirected and connected.	
\end{assumption}

\begin{assumption} \label{asp2}
	The local objective function $f_i$ is differentiable and its gradient is $w_i$-Lipschitz in $\mathbb{R}^{q}$:
	$$\lVert \nabla f_i(x)-\nabla f_i(y)\rVert \leq w_i\lVert x-y\rVert, \forall x,y\in\mathbb{R}^q, w_i>0.$$
	Denote $\overline w := \max\{w_1, \dots, w_N\}$.
\end{assumption}

\begin{assumption} \label{asp3}
	The local objective function $f_i$ is $m_i$-strongly convex:
	$$(x-y)^\top (\nabla f_i(x)-\nabla f_i(y))\ge m_i{\lVert x-y\rVert}^2, \forall x,y\in\mathbb{R}^q, m_i>0.$$
	Let $\underline m := \min\{m_1, \dots, m_N\}$.
\end{assumption}

\begin{remark}
	Because of the strong convexity of $f_i$, $f$ is strongly convex, which guarantees the uniqueness of the optimal solution to (\ref{eqfun}).
\end{remark}

\begin{assumption} \label{asp4}
	$(A_i,B_i)$ is controllable, and
	\begin{align}
		rank(C_iB_i)=q,i=1,\dots,N.\label{eqrank1}
	\end{align}
\end{assumption}

\begin{lemma}\label{lemma3}
	Under Assumption \ref{asp4}, the matrix equations
	\begin{subequations}\label{eqle}
		\begin{align}
			C_iB_iK_{\alpha_i}&=C_iA_i,\label{eqle1}\\
			C_iB_iK_{\beta_i}&=I_q,\label{eqle2}
		\end{align}
	\end{subequations}
	exist solutions $K_{\alpha_i}$, $K_{\beta_i}$.
\end{lemma}

\begin{proof}
	From (\ref{eqrank1}), we can get
	\begin{subequations}\label{eqrank2}
		\begin{align}
			rank(C_iB_i, C_iA_i)=q=rank(C_iB_i),\label{eqrank2a}\\
			rank(C_iB_i, I_q)=q=rank(C_iB_i).\label{eqrank2b}
		\end{align}
	\end{subequations}
	Thus (\ref{eqle1}) and (\ref{eqle2}) have solutions $K_{\alpha_i}$, $K_{\beta_i}$.
\end{proof}

\begin{remark}
	The controllability in Assumptions \ref{asp4} is quite standard in dealing with the problem for linear systems. And the requirement (\ref{eqrank1}) is employed to guarantee the solvability of matrix equations (\ref{eqle}), which is strictly weaker than the assumption (i.e., $rank\begin{bmatrix} C_iB_i &\boldsymbol0\\ -A_iB_i &B_i \end{bmatrix} = n_i+q$, $i=1,\dots,N$) employed in \citep{li2020,zhang2020}.
\end{remark}

\section{Main Results}\label{mai}

\subsection{Continuous Communication}

A PI controller for the $i$th agent is proposed as
\begin{subequations}\label{eqpi1}
	\begin{align}
		u_i&=-K_{\alpha_i} x_i+K_{\beta_i} (-\nabla f_i(y_i)-\sum_{j=1}^N a_{ij}(y_i-y_j)-\eta_i),\label{eqpi1a}\\
		\dot\eta_i&=\sum_{j=1}^N a_{ij}(y_i-y_j),\ \eta_i(0)=\boldsymbol 0_q,\label{eqpi1b}
	\end{align}
\end{subequations}
where $K_{\alpha_i},K_{\beta_i}$ are feedback matrices and $a_{ij}$ is the weight corresponding to the edge $(j,i)$.

\begin{figure}[htbp]
	\centering
	\includegraphics[width=0.45\textwidth]{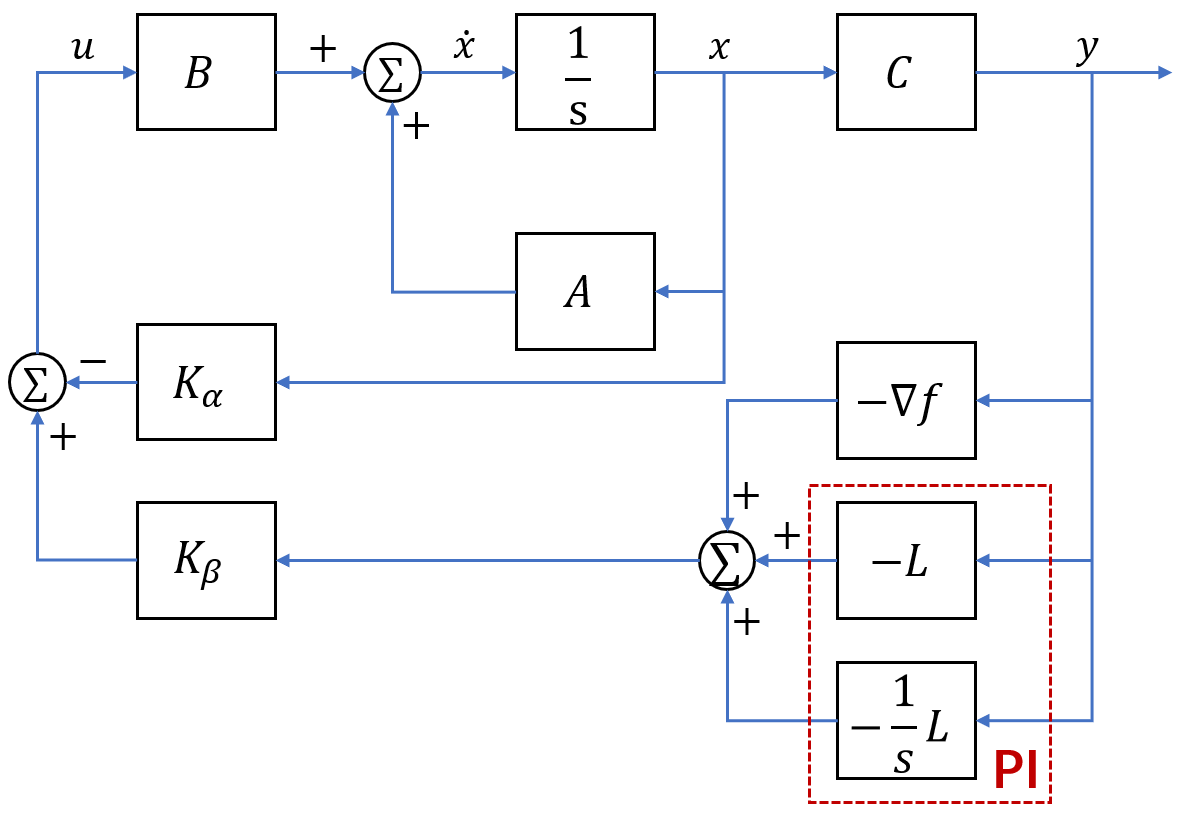} 
	\caption{Proportional-Integral controller to solve (\ref{eqfun}).}
	\label{pic1}
\end{figure}

As shown in \autoref{pic1}, the compact form of the closed-loop system is
\begin{subequations}\label{eqpi2}
	\begin{align}
		\dot x&=(A-BK_\alpha)x+BK_\beta(-\nabla f(y)-(L\otimes I_q)y-\eta),\label{eqpi2a}\\
		\dot\eta&=(L\otimes I_q)y,\ \eta(0)=\boldsymbol 0_{Nq},\label{eqpi2b}\\
		y&=Cx,\label{eqpi2c}
	\end{align}
\end{subequations}
where $x=col(x_1,\dots,x_N)$, $\eta=col(\eta_1,\dots,\eta_N)$, $y=col(y_1,\dots,y_N)$, $A=diag(A_1,\dots,A_N)$, $B=diag(B_1,\dots,B_N)$, $C=diag(C_1,\dots,C_N)$, $K_\alpha=diag(K_{\alpha_1},\dots,K_{\alpha_N})$, $K_\beta=diag(K_{\beta_1},\dots,K_{\beta_N})$, and $\nabla f(y)=col(\nabla f_1(y_1),\dots,\nabla f_N(y_N))$.
\begin{theorem}\label{th1}
	Supposing that Assumptions \ref{asp1}-\ref{asp4} hold, for linear multi-agent system (\ref{eqsys}) with control protocol (\ref{eqpi1}), the problem (\ref{eqfun}) is solved and $y_i(t)$ converges to $y^*$ exponentially as $t\to\infty$ for $i=1,\dots,N$, where $y^*\in\mathbb{R}^q$ is the optimal solution to (\ref{eqfun}), and the feedback matrices $K_{\alpha_i}$ and $K_{\beta_i}$ are solutions of equations (\ref{eqle1}) and (\ref{eqle2}), respectively. 
	
\end{theorem}

\begin{proof}
	First, we discuss the relationship between the equilibrium point of (\ref{eqpi2}) and the optimal solution of (\ref{eqfun}). Pre-multiplying (\ref{eqpi2a}) by $C$, we get
	\begin{subequations}\label{eqpi3}
		\begin{align}
			\dot y&=-\nabla f(y)-(L\otimes I_q)y-\eta,\label{eqpi3a}\\
			\dot\eta&=(L\otimes I_q)y,\ \eta(0)=\boldsymbol 0_{Nq}.\label{eqpi3b}
		\end{align}
	\end{subequations}
	
	Because $\eta(0)=\boldsymbol 0_{Nq}$ and $\boldsymbol 1_N^\top L=\boldsymbol0_N^\top $, we can get $(\boldsymbol 1_N^\top \otimes I_q)\eta(t)=0,\forall t>0$. Pre-multiply (\ref{eqpi3a}) by $(\boldsymbol 1_N^\top \otimes I_q)$, and let the derivatives of (\ref{eqpi3}) be equal to $\boldsymbol 0$. Then the equilibrium point $(\bar y,\bar\eta)$ satisfies
	\begin{subequations}
		\begin{align}
			\sum_{i=1}^N \nabla f_i(\bar y_i)=0,\label{eqbala}\\
			(L\otimes I_q)\bar y=\boldsymbol 0_{Nq},\label{eqbalb}
		\end{align}
	\end{subequations}
	where $\bar y=col(\bar y_1, \dots, \bar y_N)$.
	
	From (\ref{eqbalb}) we know all $\bar{y}_i$ reach consensus, which together with (\ref{eqbala}) ensures that the equilibrium point is the optimal solution to (\ref{eqfun}).
	
	To proceed, taking $\rho=y-\bar y$, $\sigma=\eta-\bar\eta$, one has
	\begin{subequations}\label{eqpi4}
		\begin{align}
			\dot{\rho}&=-\boldsymbol h-(L\otimes I_q)\rho-\sigma,\label{eqpi4a}\\
			\dot{\sigma}&=(L\otimes I_q)\rho\label{eqpi4b},
		\end{align}
	\end{subequations}
	where $\boldsymbol h:=\nabla f(y)-\nabla f(\bar y)$.
	
	Next we only need to discuss the convergence of (\ref{eqpi4}).

	Select a Lyapunov candidate as
	\begin{align}
		V =&\frac{\xi}{2}\rho^\top \rho + \sigma^\top \rho + \frac{1}{2}\sigma^\top \sigma + \frac{\xi}{2}\sigma^\top (\Gamma\otimes I_q)\sigma\notag\\
		\le&\overline\lambda_{_E} p^\top p,\label{v1}
	\end{align}
	where $p:=col(\rho,\sigma)$, $\xi$ is a parameter to be determined, $\Gamma$ is given in Lemma \autoref{lemma2}, and $\overline\lambda_{_E}$ is the maximum eigenvalue of
	\begin{align}
		E:=\frac{1}{2}
		\begin{pmatrix}
			\xi I_{Nq}& I_{Nq}& \\
			I_{Nq}& I_{Nq}+\xi (\Gamma\otimes I_q)& \\
		\end{pmatrix}.\notag
	\end{align}
	It is easy to verify that $E$ is a positive definite matrix for $\xi \ge 1$ by Schur complement.
	
	The derivative of (\ref{v1}) is
	\begin{align}
		\dot V = &-\xi\rho^\top \boldsymbol h - (\xi-1)\rho^\top (L\otimes I_q)\rho - \sigma^\top \sigma - \sigma^\top \boldsymbol h\notag\\
		& - \xi\rho^\top ((I_N-\Pi)\otimes I_q)\sigma,\label{v1'}
	\end{align}
	where $\Pi$ is defined in Lemma \autoref{lemma2}.
	
	By the strong convexity of $f_i$, we have
	\begin{align}
		\rho^\top \boldsymbol h\ge \underline m\rho^\top \rho.\label{zm1}
	\end{align}
	
	Because the gradients of local objective functions are Lipschitz, we have $
	\left \| \boldsymbol h \right \|\le \overline w\left\|\rho\right\|$.
	Then one has
	\begin{align}
		-\sigma^\top \boldsymbol h \le \frac{1}{2}\sigma^\top \sigma + \frac{1}{2}\overline w^2\rho^\top \rho.\label{zm2}
	\end{align}
	
	From the fact that $(\boldsymbol 1_N^\top \otimes I_q)\eta=0$, it can be obtained that $((I_N-\Pi)\otimes I_q)\sigma=0$. 
	Since the matrix $L\ge 0$ and $\xi\ge1$, we can get $(\xi-1)\rho^\top (L\otimes I_q)\rho\ge 0$. Substituting (\ref{zm1}), (\ref{zm2}) into (\ref{v1'}), one has
	\begin{align}
		\dot V \le &-\xi \underline m\rho^\top \rho + \frac{1}{2}\overline w^2\rho^\top \rho - \frac{1}{2}\sigma^\top \sigma\notag\\
		= & -p^\top F_1p,\notag
	\end{align}
	
	with
	\begin{align}
		F_1:=
		\begin{pmatrix}
			(\xi\underline m-\frac{1}{2}\overline w^2)\otimes I_{Nq}& \boldsymbol0_{Nq}& \\
			\boldsymbol0_{Nq}& \frac{1}{2} I_{Nq}& \\
		\end{pmatrix}.\notag
	\end{align}
	
	For $\xi > \frac{\overline w^2}{2\underline m}$, the matrix $F_1$ is positive definite. Note that $\underline \lambda_{F_1}:=\min\{\xi\underline m-\frac{1}{2}\overline w^2, \frac{1}{2}\}$ is the smallest eigenvalues of $F_1$,
	\begin{align}
		\dot V \le -\underline \lambda_{F_1}p^\top p.\label{v1'1}
	\end{align}
	Finally, setting $\xi>\max\{1,\frac{\overline w^2}{2\underline m}\}$, by Theorem 4.10 in \citet{nonlinear}, we can conclude the global exponential stability of system (\ref{eqpi3}), and the variable $y$ satisfies
	\begin{align}
		\|y(t)-y^*\| \le c_1e^{-\frac{c_2}{2}t},\label{yt}
	\end{align}
	where $c_1>0$ is some constant and $c_2:=\frac{\underline \lambda_{F_1}}{\overline\lambda_{_E}}$.	
\end{proof}

\begin{remark}
	For each agent, the parameters in algorithm (\ref{eqpi1}) only depend on its own information, so the proposed algorithm is fully distributed. In comparison, an exponential convergence rate is established, while the most related work \citep{li2020} only provides an asymptotic convergence without analysis of the convergence speed building upon a stronger assumption than Assumption \ref{asp4} here. Meanwhile, compared with \citet{zhang2020}, where an exponential convergence rate is obtained for homogeneous linear multi-agent systems under fixed directed graphs, the exponential convergence is established here for heterogeneous linear multi-agent systems based on a strictly weaker Assumption \ref{asp4}, including homogeneous linear multi-agent systems as a special case. It should be also noted that the algorithm in Theorem \ref{th1} can be extended to fixed directed graphs just like \citet{zhang2020}.
\end{remark}

\begin{remark}
	For the convergence rate, $c_2$ reaches the maximum value
	$\overline c_2:=\frac{2}{\xi+\frac{\xi}{\lambda_2} +1 + \sqrt{(\frac{1}{\lambda_2^2}-\frac{2}{\lambda_2}+1)\xi^2+(\frac{2}{\lambda_2}-2)\xi+5}}$ when choosing $\xi=\frac{\overline w^2 + 1}{2\underline m}$ in (\ref{yt}), where $\lambda_2$ is the second smallest eigenvalue of $L$, called the algebraic connectivity of $\mathcal{G}$.
	To compare with the algorithm in \citet{kia2015} for first-order integrator systems, by setting global parameters $\alpha=\beta=1$ there, its variable $x$ has a similar convergence to (\ref{yt}):  $\|x(t)-x^*\| \le c_3e^{-\frac{c_4}{2}t}$, where $c_3>0$ is some constant and $c_4:=\frac{4min\{(\phi+1)\underline m-\frac{1}{2}\overline w^2,\frac{1}{2}\}}{\phi+\frac{\phi}{\lambda_2} + \frac{1}{\lambda_2} +2 + \sqrt{(\frac{1}{\lambda_2^2}-\frac{2}{\lambda_2}+1)\phi^2+(\frac{2}{\lambda_2^2}-\frac{2}{\lambda_2})\phi+\frac{1}{\lambda_2^2}+4}}$
	with a parameter $\phi>\max\{1,\frac{\overline w^2}{2\underline m}-1\}$. 
	When $\overline w^2\ge4\underline m-1$, by choosing $\phi=\frac{\overline w^2 - 2\underline m + 1}{2\underline m}$, $c_4$ achieves its maximum value which is the same as $\overline c_2$, that is, the same convergence rate is obtained for both algorithms.
	When $\overline w^2<4\underline m-1$, $c_4$ reaches the maximum value $\overline c_4:=\frac{2}
	{\frac{2}{\lambda_2}+3+\sqrt{\frac{4}{\lambda_2^2}-\frac{4}{\lambda_2}+5}}$ 
	by choosing $\phi=1$, and it is easy to verify that $\overline c_2>\overline c_4$, which shows that the established convergence rate for algorithm (\ref{eqpi1}) in this paper is faster than that of the algorithm in \citet{kia2015}. 
\end{remark}

\subsection{Periodic Communication}
In order to avoid continuous communication and reduce communication overhead, we next discuss the case of discrete communication.

Suppose that $t_k^i$ is the $k$th communication instant of the $i$th agent, and denote $\hat{y}_i(t):=\hat y_i(t_k^i), \forall t\in [t_k^i,t_{k+1}^i)$ as the latest known output of agent $i\in V$ transmitted to its neighbors. The communication instant sequence of the $i$th agent $\{t_1^i, \dots, t_k^i, \dots\}$ will be determined later. We define a measurement error $e_i:=\hat y_i(t)-y_i(t)$, and it is clear that $e_i=0$ at any instant $t_k^i$.

Consider the next implementation of the algorithm (\ref{eqpi1}) with discrete-time communication,
\begin{subequations}\label{eqpib}
	\begin{align}
		u_i&=-K_{\alpha_i}x_i+K_{\beta_i}(-\nabla f_i(y_i)-\sum_{j=1}^N a_{ij}(\hat y_i-\hat y_j)-\eta_i),\label{eqpiba}\\
		\dot\eta_i&=\sum_{j=1}^N a_{ij}(\hat y_i-\hat y_j),\ \eta_i(0)=\boldsymbol 0_q,\label{eqpibb}
	\end{align}
\end{subequations}
where $K_{\alpha_i},K_{\beta_i}$ are feedback matrices and $a_{ij}$ is the weight corresponding to the edge $(j,i)$.

The following is the conclusion of the periodic communication control law.

\begin{theorem}\label{th2}
	Supposing that Assumptions \ref{asp1}-\ref{asp4} hold, for linear multi-agent system (\ref{eqsys}) with discrete control protocol (\ref{eqpib}), the problem (\ref{eqfun}) is solved and $y_i(t)$ converges to $y^*$ exponentially as $t\to\infty$ for $i=1,\dots,N$, if the communication instant is set as $t_k^{i+1}=t_k^i+ \Delta$, $\forall\Delta \in (0, \tau_0]$, where $y^*\in\mathbb{R}^q$ is the optimal solution to (\ref{eqfun}),
	\begin{align}
		\tau_0:=\frac{1}{\overline{w}+1} \ln{\bigg(1+ \frac{(\overline{w}+1)\epsilon} {\overline{w}+1+\sqrt{2}\lambda_N+\sqrt{2}\lambda_N\epsilon}\bigg)},\label{tau1}
	\end{align}
	with $\epsilon:=\frac{1}{2\sqrt{2(\xi^2+(\xi-1)^2)}}$ and $\xi>\frac{4\overline{w}^2+2\lambda_N^2+1}{8\underline m}$, and the feedback matrices $K_{\alpha_i}$ and $K_{\beta_i}$ are solutions of equations (\ref{eqle1}) and (\ref{eqle2}), respectively.
\end{theorem}

\begin{proof}
	The compact form of the closed-loop system is
	\begin{subequations}\label{eqpib2}
		\begin{align}
			\dot x&=(A-BK_\alpha)x+BK_\beta(-\nabla f(y)-(L\otimes I_q)\hat y-\eta),\label{eqpib2a}\\
			\dot\eta&=(L\otimes I_q)\hat y,\ \eta(0)=\boldsymbol 0_{Nq},\label{eqpib2b}\\
			y&=Cx,\label{eqpib2c}
		\end{align}
	\end{subequations}
	where $\hat y=col(\hat y_1,\dots,\hat y_N)$.
	
	Using the same state transformation as in the previous section and letting $\hat\rho:=\hat y-\bar{y}$, the dynamics (\ref{eqpib2}) can be written as
	\begin{subequations}\label{eqpib3}
		\begin{align}
			\dot{\rho}&=-\boldsymbol h-(L\otimes I_q)\hat\rho-\sigma,\label{eqpib3a}\\
			\dot{\sigma}&=(L\otimes I_q)\hat\rho\label{eqpib3b}.
		\end{align}
	\end{subequations}
	And because $e_i=\hat y_i(t)-y_i(t)$, we have $\hat\rho=\rho+e$ with $e=col(e_1, \dots, e_N)$:
	\begin{subequations}\label{eq1pib4}
		\begin{align}
			\dot{\rho}&=-\boldsymbol h-(L\otimes I_q)(\rho+e)-\sigma,\label{eqpib4a}\\
			\dot{\sigma}&=(L\otimes I_q)(\rho+e)\label{eqpib4b}.
		\end{align}
	\end{subequations}
	
	Selecting the Lyapunov candidate as in (\ref{v1}), its derivative is
	\begin{align}
		\dot V = &-\xi\rho^\top \boldsymbol h - (\xi-1)\rho^\top (L\otimes I_q)\rho - \sigma^\top \sigma - \sigma^\top \boldsymbol h\notag\\
		&-(\xi-1)\rho^\top (L\otimes I_q)e+\xi\sigma^\top (\Pi\otimes I_q)e,\label{v2'}
	\end{align}
	where $\Pi$ is defined in Lemma \autoref{lemma2}.
	
	By the inequality $-(\xi-1)\rho^\top (L\otimes I_q)e
	\le \frac{\lambda_N^2}{4}\rho^\top \rho+(\xi-1)^2e^\top e$
	and
	$\xi\sigma^\top (\Pi\otimes I_q)e
	\le\frac{1}{4}\sigma^\top \sigma + \xi^2e^\top e$, 
	one has
	\begin{align}
		\dot V \le &-\xi\underline m \rho^\top \rho - \frac{1}{2}\sigma^\top \sigma + \frac{1}{2}\overline w^2\rho^\top \rho + \frac{\lambda_N^2}{4}\rho^\top \rho \notag\\
		&+(\xi-1)^2e^\top e + \frac{1}{4}\sigma^\top \sigma + \xi^2e^\top e \notag\\
		=&-p^\top F_2p-\frac{1}{8}(p^\top p-\epsilon^{-2}e^\top e),\label{v2'2}
	\end{align}
	where $p=col(\rho,\sigma)$, $\epsilon=\frac{1}{2\sqrt{2(\xi^2+(\xi-1)^2)}}$ and
	\begin{align}
		F_2:=
		\begin{pmatrix}
			(\xi\underline m-\frac{1}{2}\overline w^2-\frac{\lambda_N^2}{4}-\frac{1}{8})\otimes I_{Nq}& \boldsymbol0_{Nq}& \\
			\boldsymbol0_{Nq}& \frac{1}{8} I_{Nq}& \\
		\end{pmatrix}.\notag
	\end{align}
	
	For $\xi > \frac{4\overline w^2+2\lambda_N^2+1}{8\underline m}$, the matrix $F_2$ is positive definite.
	Inspired by \citet{tran2018}, let $q=\frac{\|e\|}{\|p\|}$, and then its derivative is
	\begin{align}
		\dot q = &\frac{\dot{\| e\|}\|p\|-\|e\|\dot{\| p\|}}{\|p\|^2} 
		\le\frac{(1+q)\dot{\| p\|}}{\|p\|}. \notag
	\end{align}
	
	Because
	\begin{align}
		\dot{\| p\|}\le &\overline w\|\rho\|+\|\sigma\|+\sqrt 2\lambda_N\|\rho\|+\sqrt 2\lambda_N\|e\| \notag\\
		\le&(\overline w + 1 + \sqrt 2\lambda_N)\|p\| + \sqrt 2\lambda_N\|e\|,\notag
	\end{align}
	we can conclude
	\begin{align}
		\dot q \le &(\overline w + 1 + \sqrt 2\lambda_N)(1+q) + \sqrt 2\lambda_N(1+q)q \notag\\
		=&(\overline w + 1)(1+q) + \sqrt 2\lambda_N(1+q)^2.\notag
	\end{align}
	
	By $e(t_k)=0$, one has $q(t_k)=0$. Choosing the differential equation
	\begin{align}
		\dot \mu = (\overline w + 1)(1+\mu) + \sqrt 2\lambda_N(1+\mu)^2,\label{diu}
	\end{align}
	the solution $\mu (t)$ of (\ref{diu}) with $\mu (0)=0$ is
	\begin{align}
		\mu (t) = \frac{(e^{(\overline w + 1)t}-1)(\overline w + 1 + \sqrt 2\lambda_N)}{\overline w + 1 + \sqrt 2\lambda_N - \sqrt 2\lambda_Ne^{(\overline w + 1)t}}.
	\end{align}
	
	Based on the  Lemma 3.4 in \citet{nonlinear}, the solution $q(t)$ satisfies $q(t)\le\mu(t), \forall t>0$. Thus, for $t_{k+1}-t_k\le\tau_0$, we have $q(t)\le\mu(\tau_0)=\epsilon$, in which $\tau_0$ is defined in (\ref{tau1}). Therefore, we get $p^\top p-\epsilon^{-2}e^\top e\ge0$.
	
	Note that $\underline \lambda_{F_2}$ is the smallest eigenvalues of $F_2$:
	\begin{align}
		\dot V \le -\underline \lambda_{F_2}p^\top p.
	\end{align}
	
	Finally, setting $\xi>\max\{1,\frac{4\overline w^2+2\lambda_N^2+1}{8\underline m}\}$, by Theorem 4.10 in \citet{nonlinear}, we can conclude the global exponential stability of system (\ref{eqpib2}).
\end{proof}

\subsection{Event-triggered Communication}

When the outputs of agents do not change too much, periodic communication schemes will transmit a lot of unnecessary data. So we introduce an event-triggered mechanism to further reduce the communication overhead.
\begin{theorem}\label{th3}
	Supposing that Assumptions \ref{asp1}-\ref{asp4} hold, for linear multi-agent system (\ref{eqsys}) with discrete control protocol (\ref{eqpib}), the problem (\ref{eqfun}) is solved and $y_i(t)$ converges to $y^*$ exponentially as $t\to\infty$ for $i=1,\dots,N$, if the communication instant is chosen as $t_k^{i+1}=t_k^i+ \max\{\tau_k^i, \Delta\}$, where
	\begin{align}
		\tau_k^i:=\inf_{t>t_i^k}\{t-t_i^k|\|e_i(t)\|^2=\frac{1}{4(d^{out}_i+\kappa)}\sum_{j=1}^{N} a_{ij}\|\hat y_i - \hat y_j \|^2 \},\label{tau2}
	\end{align}
	with the parameter $\kappa>\max\{\frac{\overline{w}^2}{4\underline m}, \frac{1}{2}\}$, and $y^*, K_{\alpha_i}, K_{\beta_i}, \Delta$ are the same as in \autoref{th2}.
\end{theorem}

\begin{proof}
	Selecting the Lyapunov candidate (\ref{v1}), its derivative can be calculated as in (\ref{v2'}).
	
	Carrying out a transformation as in the proof of \autoref{th2}, one can get
	\begin{align}
		\xi\sigma^\top (\Pi\otimes I_q)e \le \frac{\xi^2}{4(\xi-1)\kappa}\sigma^\top \sigma + (\xi-1)\kappa e^\top e,\label{zm5}
	\end{align}
	in which $\kappa\in \mathbb{R}$ is a parameter to be determined later.

	Substituting (\ref{zm1}), (\ref{zm2}), (\ref{zm5}) into (\ref{v2'}) and noting $p=col(\rho,\sigma)$, one has
	\begin{align}
		\dot V \le &-\xi\underline m\rho^\top \rho + \frac{\overline{w}^2}{2}\rho^\top \rho - \frac{1}{2}\sigma^\top \sigma  - \frac{\xi-1}{2}\rho^\top (L\otimes I_q)\rho\notag\\
		&- (\xi-1)\rho^\top (L\otimes I_q)e + \frac{\xi^2}{4(\xi-1)\kappa}\sigma^\top \sigma  + (\xi-1)\kappa e^\top e\notag\\
		\le&-p^\top F_3p-(\xi-1)\bar s,\notag
	\end{align}
	where 
	\begin{align}
		F_3:=
		\begin{pmatrix}
			(\xi\underline m-\frac{1}{2}\overline w^2)\otimes I_{Nq}& \boldsymbol0_{Nq}& \\
			\boldsymbol0_{Nq}& (\frac{1}{2}-\frac{\xi^2}{4(\xi-1)\kappa})\otimes I_{Nq}&\\
		\end{pmatrix},\notag
	\end{align}
	$\bar s := s - \kappa e^\top e$, and $s=\frac{1}{2}\rho^\top (L\otimes I_q)\rho+\rho^\top (L\otimes I_q)e$.
	
	By using $\boldsymbol1_N^\top L=\boldsymbol0_N^\top$ and $\rho=y-\bar y$, we get
	\begin{align}
		s =&\frac{1}{2}(y-\bar y)^\top (L\otimes I_q)(y-\bar y)+(y-\bar y)^\top (L\otimes I_q)e\notag\\
		=&\frac{1}{2}y^\top (L\otimes I_q)y+y^\top (L\otimes I_q)e\notag\\
		=&\frac{1}{2}\hat y^\top (L\otimes I_q)\hat y - \frac{1}{2}e^\top (L\otimes I_q)e. 
	\end{align}
	
	Because $L=D^{out}-A$ and $D^{out}+A\ge0$, we can get $\frac{1}{2}e^\top (L\otimes I_q)e \le e^\top (D^{out}\otimes I_q)e = \sum_{i=1}^{N}d^{out}_ie_i^\top e_i$.
	
	Because $\sum_{i=1}^{N}\sum_{j=1}^{N}a_{ij}(\hat y_i^\top \hat y_i - \hat y_j^\top \hat y_j)=0$, one has
	\begin{align}
		&\hat y^\top (L\otimes I_q)\hat y\notag\\
		=&\sum_{i=1}^{N}\sum_{j=1}^{N}a_{ij}\hat y_i^\top (\hat y_i-\hat y_j) - \frac{1}{2}\sum_{i=1}^{N} \sum_{j=1}^{N}a_{ij}(\hat y_i^\top \hat y_i - \hat y_j^\top \hat y_j)\notag\\
		=&\frac{1}{2}\sum_{i=1}^{N}\sum_{j=1}^{N} a_{ij}\left\|\hat y_i - \hat y_j \right\|^2.
	\end{align}
	Therefore, it can be obtained that
	\begin{align}
		s\ge \frac{1}{4}\sum_{i=1}^{N}\sum_{j=1}^{N} a_{ij}\left\|\hat y_i - \hat y_j \right\|^2 - \sum_{i=1}^{N}d^{out}_ie_i^\top e_i.
	\end{align}
	
	By the triggering conditions (\ref{tau2}), we get
	\begin{align}
		\bar s\ge &\sum_{i=1}^{N} \bigg(\frac{1}{4}\sum_{j=1}^{N} a_{ij}\|\hat y_i - \hat y_j \|^2 - (d^{out}_i+\kappa)\|e_i\|^2 \bigg)\notag\\
		\ge &0,
	\end{align}
	thereby implying
	\begin{align}
		\dot{V}\le -p^\top F_3p.
	\end{align}
	
	Choosing $\kappa>\max\{\frac{\overline{w}^2}{4\underline m}, \frac{1}{2}\}$, there must be a parameter $\xi>1$ that makes the matrix $F_3$ positive definite. Note that $\underline \lambda_{F_3}$ is the smallest eigenvalues of $F_3$:
	\begin{align}
		\dot V \le -\underline \lambda_{F_3}p^\top p.
	\end{align}
	
	By the Theorem 4.10 in \citet{nonlinear}, we can conclude the global exponential stability of this system.
\end{proof}

\begin{remark}
	In \autoref{th2} and \autoref{th3}, we have designed the discrete-time communication algorithms for the optimal output consensus of continuous heterogeneous linear multi-agent systems. Compared with the most related work \citep{li2020} which only proves asymptotic convergence and has gradually decreasing communication interval, our algorithms guarantee exponential convergence and clearly give a lower bound of the communication interval, thus excluding the Zeno behavior.
\end{remark}
\section{Simulation}\label{sim}


\begin{figure*}[!t]\centering
	\subfigure[Continuous Communication.] {\includegraphics[height=1.65in,width=2.3in,angle=0]{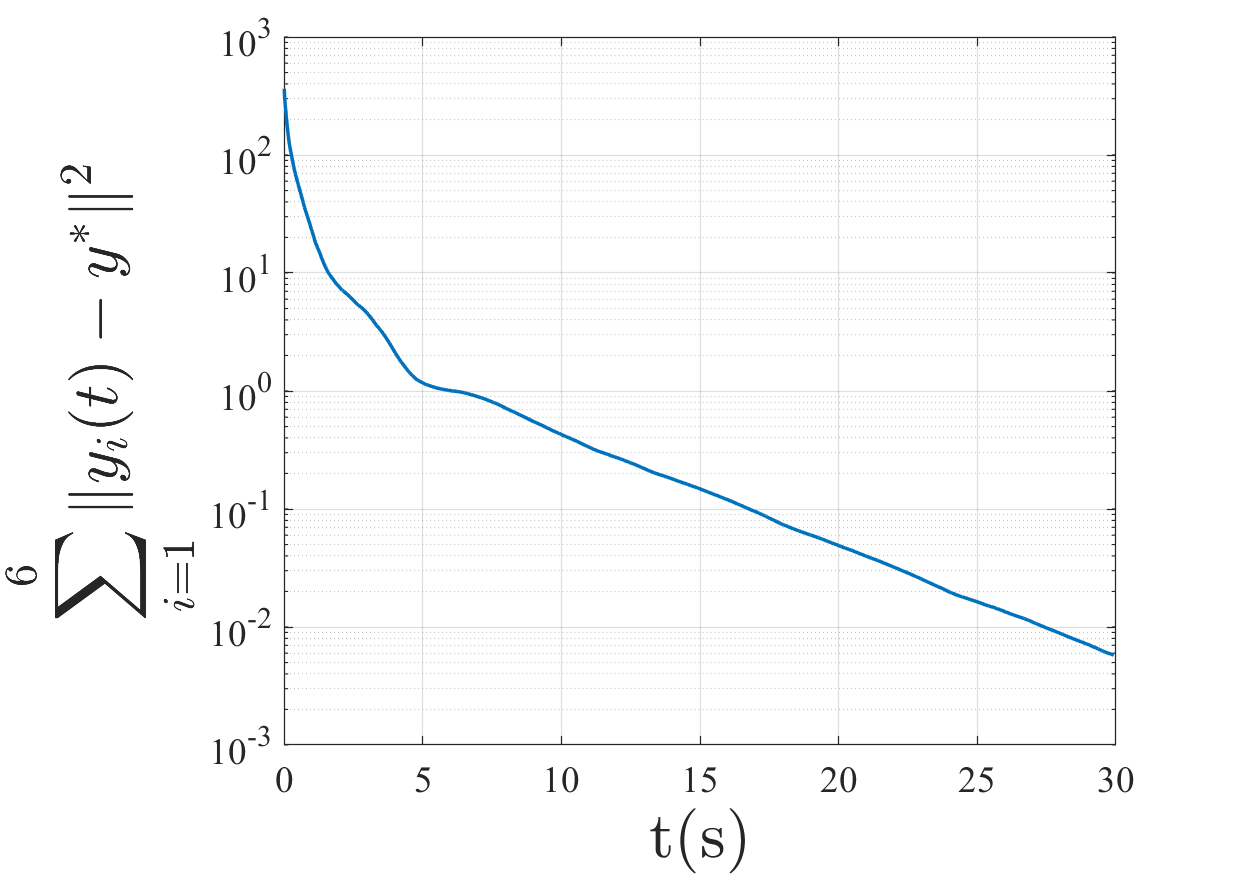}}\label{fig21}
	\subfigure[Periodic Communication.] {\includegraphics[height=1.65in,width=2.3in,angle=0]{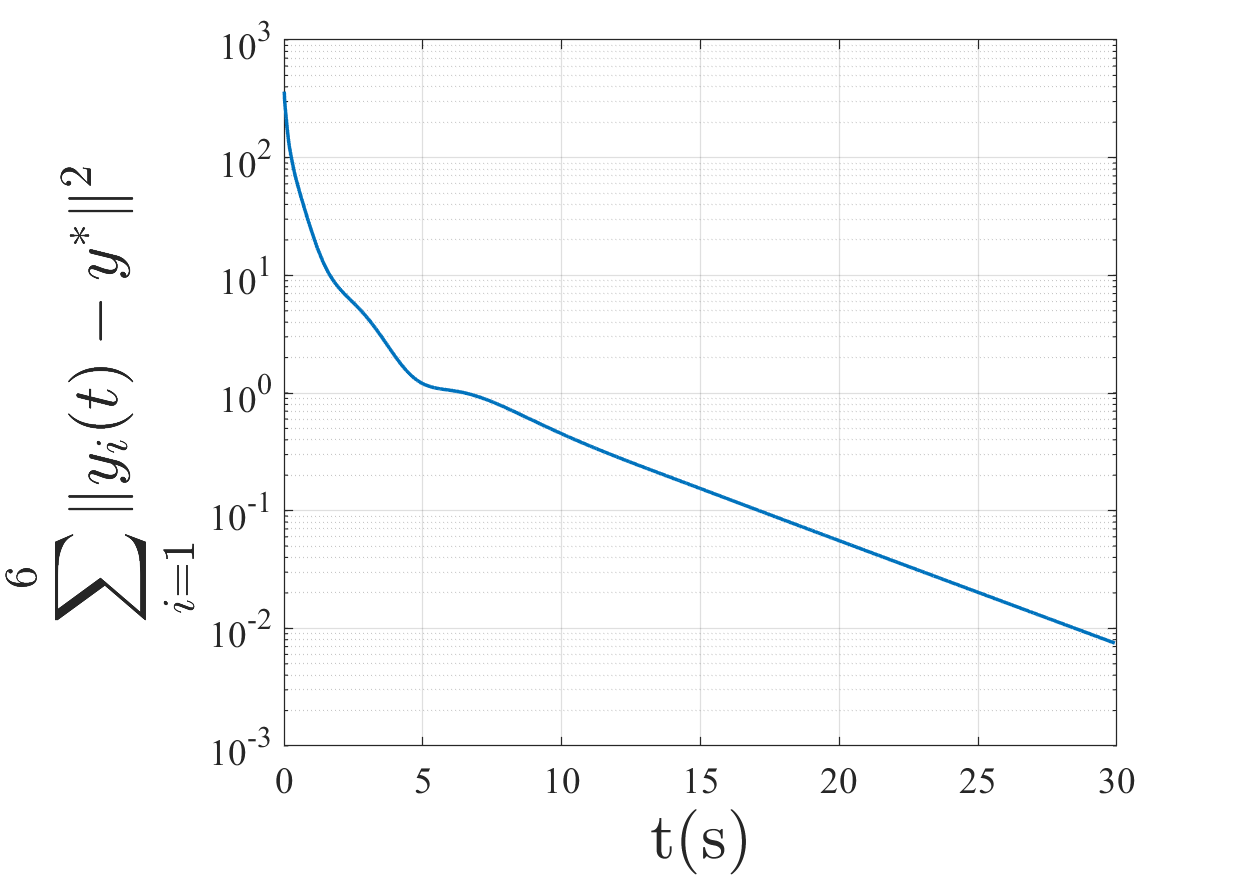}}
	\subfigure[Event-triggered Communication.] {\includegraphics[height=1.65in,width=2.3in,angle=0]{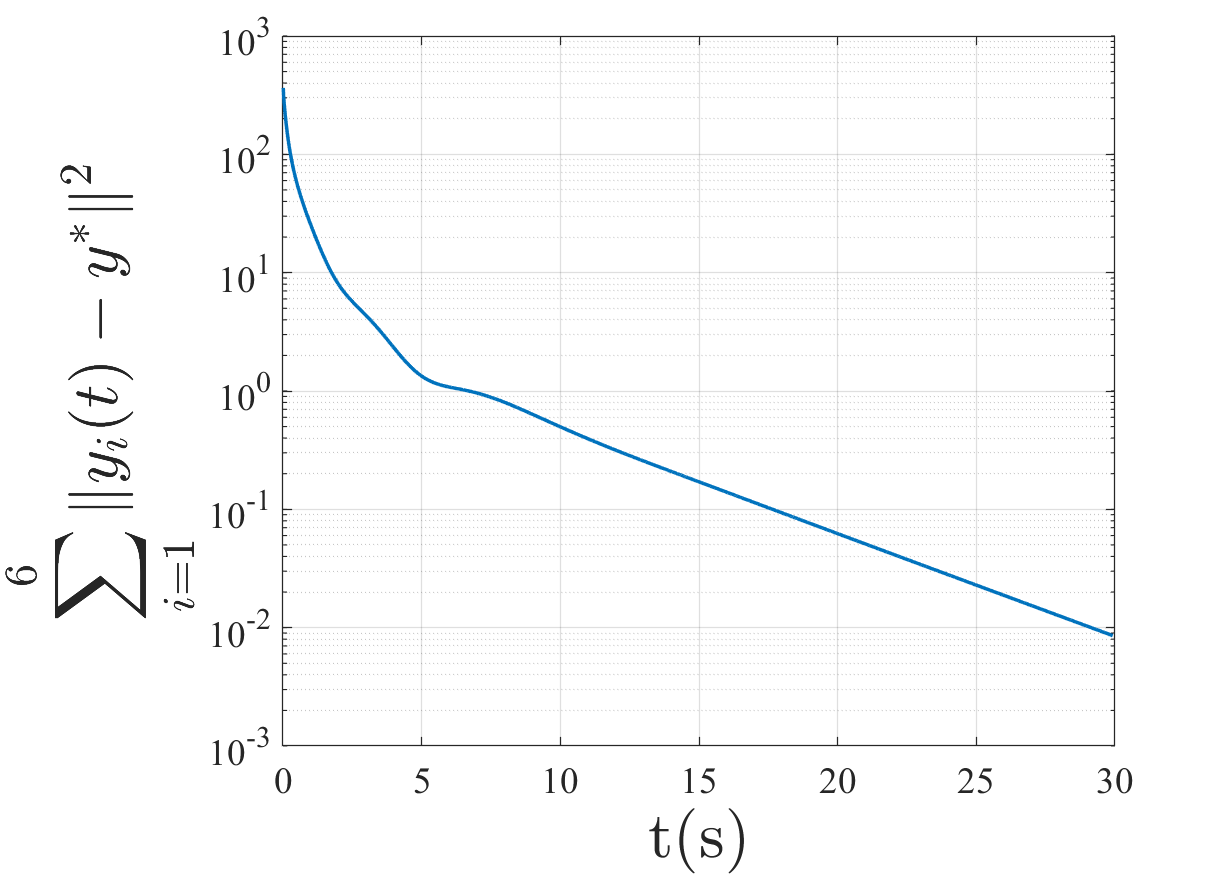}}
	\caption{ The outputs error $\sum_{i=1}^{6}\|y_i(t)-y^*\|^2$.}\label{pic3}	
\end{figure*}

\begin{figure*}[!t]\centering
	\subfigure[Periodic Communication.] {\includegraphics[height=1.8in,width=2.85in,angle=0]{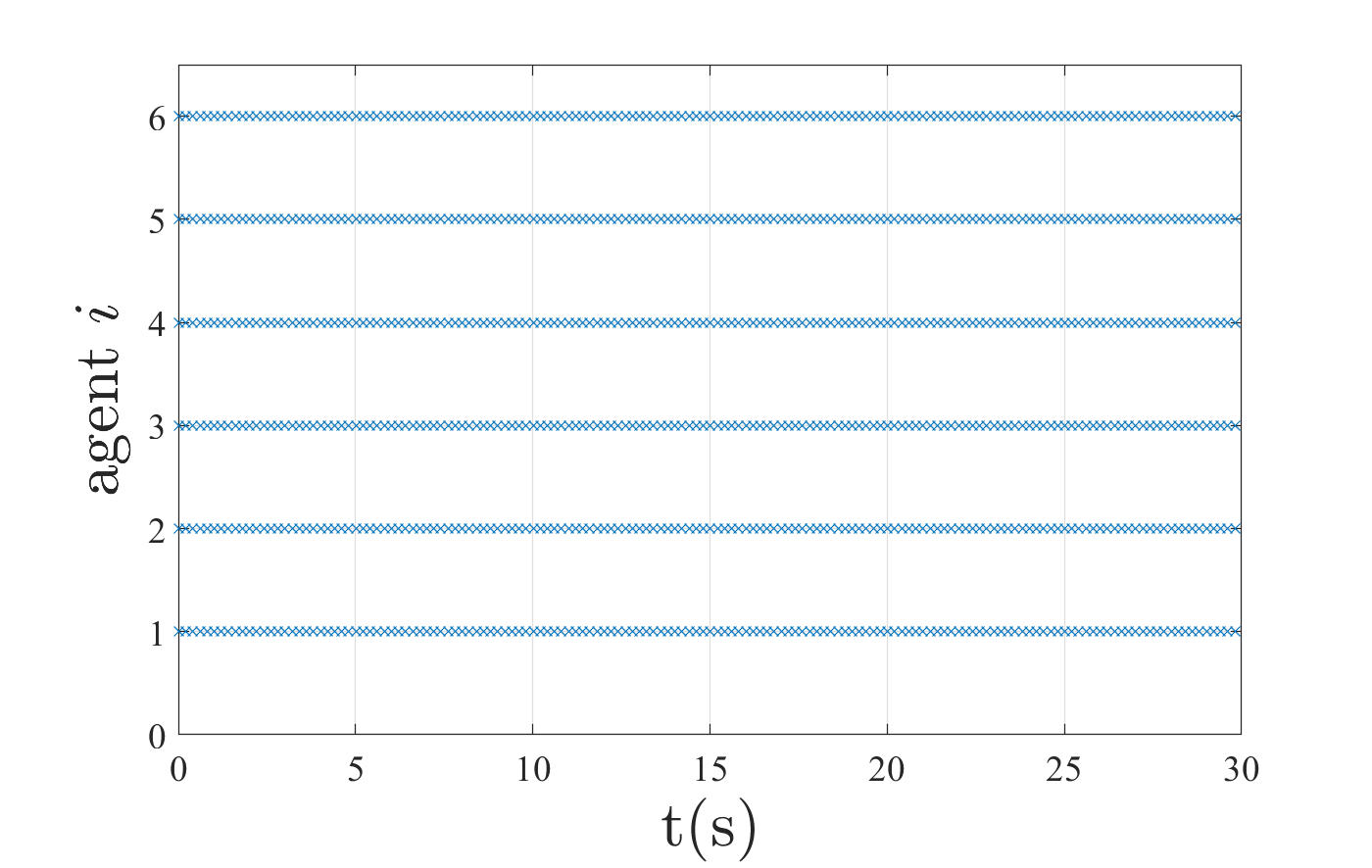}}
	\subfigure[Event-triggered Communication.] {\includegraphics[height=1.8in,width=2.85in,angle=0]{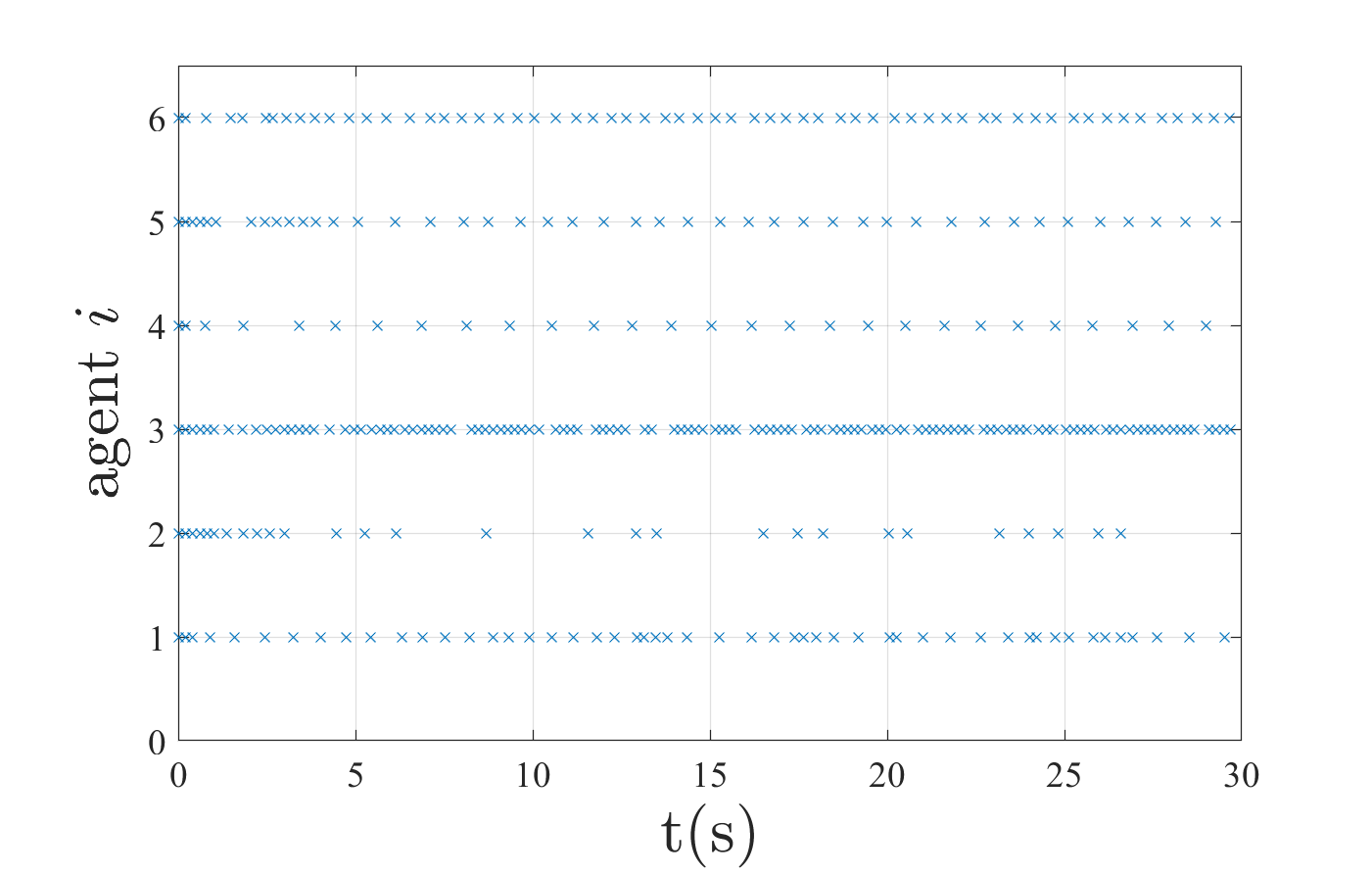}}
	\caption{ Triggering instants of six agents. }\label{pic4}
\end{figure*}

\begin{example}
	Consider a network of six agents, where $A_{1,2}$ $=$ $[1,0;0,1]$, $A_{3,4}$ $=$ $[0,1;-2,1]$, $A_{5,6}$ $=$ $[1,1,0;0,1,1;1,0,1]$, $B_{1,2}$ $=$ $[0,1;1,-2]$, $B_{3,4}$ $=$ $[1,1;1,0]$, $B_{5.6}$ $=$ $[1,0;0,1;2,0]$, $C_{1,2}$ $=$ $[3,0;0,1]$, $C_{3,4}$ $=$ $[2,2;-1,1]$, $C_{5,6}$ $=$ $[1,-1,2;1,2,2]$. The local objective functions are as follows
	with   decision variable $y=(y_a,y_b)^\top \in \mathbb{R}^2$ :
	\begin{equation}\label{sf}
	\begin{array}{lll}
	f_1(y)=(y_a-5)^2+2(y_b-3)^2;\\
	f_2(y)=(2y_a+5y_b-9)^2+\frac{0.2y_a^2}{\sqrt{2y_a^2+2}};\\
	f_3(y)=\ln(e^{0.1y_a}+e^{0.1y_b});\\
	f_4(y)=(2y_a+1)^2+2(y_b-1)^2;\\
	f_5(y)=(y_a+y_b)^2+\ln(y_b+3);\\
	f_6(y)=\lVert y\rVert^2+y_a+y_b.\notag
	\end{array}
	\end{equation}
	
	The communication network among these agents is depicted as \autoref{pic6} with all the edge weights as $1$.

	\begin{figure}[htbp]
		\centering
		\includegraphics[width=0.2\textwidth]{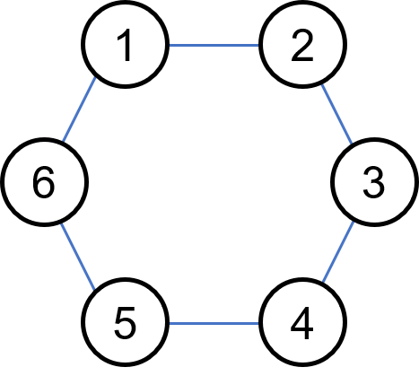} 
		\caption{Communication network among six agents.}
		\label{pic6}
	\end{figure}
	
	It can be verified that Assumptions \ref{asp1}$-$\ref{asp4} hold. And we can calculate $y^*=(0.26224, 1.59614)^\top $ by minimizing the global cost function $f(y)=\sum_{i=1}^{6}f_i(y)$.
	
	The parameters of each agent can be selected by the proposed algorithms, where $K_{\alpha_{1,2}}$ $=$ $[2,1;1,0]$, $K_{\alpha_{3,4}}$ $=$ $[2,-1;-2,2]$, $K_{\alpha_{5,6}}$ $=$ $[0.6,0.2,0.4;0,1,1]$, $K_{\beta_{1,2}}$ $=$ $[0.667,1;0.333,0]$, $K_{\beta_{3,4}}$ $=$ $[0.25,0.5;0,-1]$, $K_{\beta_{5,6}}$ $=$ $[0.133,0.0667;-0.333,0.333]$, and $\Delta=0.2$. The initial values $x_i(0)$ are randomly selected in $[-10, 10]$.
	
	\autoref{pic3} depicts the optimization errors $\sum_{i=1}^{6}\|y_i(t)-y^*\|^2$ with three control laws of continuous, periodic, and event-triggered communication respectively. It can be seen that the outputs of all agents converge to the optimal value $y^*$ exponentially.
	\autoref{pic4} shows the triggering instants of six agents with periodic and event-triggered communication control laws, from which we can observe that the communication among six agents is discrete and neither of them exhibits Zeno behavior. Compared with the periodic communication control law, the event-triggered mechanism can further reduce communication overhead.
\end{example}

\begin{example}
	In order to verify the convergence speed of our algorithm, we compare it with the most related work \citep{li2020} in the case of continuous communication. For convenience, we adopt the linear system parameters, objective functions and communication network that are consistent with the simulation in \citet{li2020}. The parameters of our algorithm are selected as $K_{\alpha_{1,2}}$ $=$ $[0,1;0,0]$, $K_{\alpha_{3,4}}$ $=$ $[0.5,-0.5;0,0]$, $K_{\alpha_{5,6}}$ $=$ $[0.25,1,-1.5;0,0,0]$, $K_{\beta_{1,2}}$ $=$ $[1;0]$, $K_{\beta_{3,4}}$ $=$ $[0.5;0]$, $K_{\beta_{5,6}}$ $=$ $[0.5;0]$. And the parameters of the algorithm in \citet{li2020} are selected to be the default values there. The initial values $x_i(0)$ are randomly selected in $[-50, 50]$.
	
	\begin{figure}[htbp]
		\centering
		\includegraphics[width=0.4\textwidth]{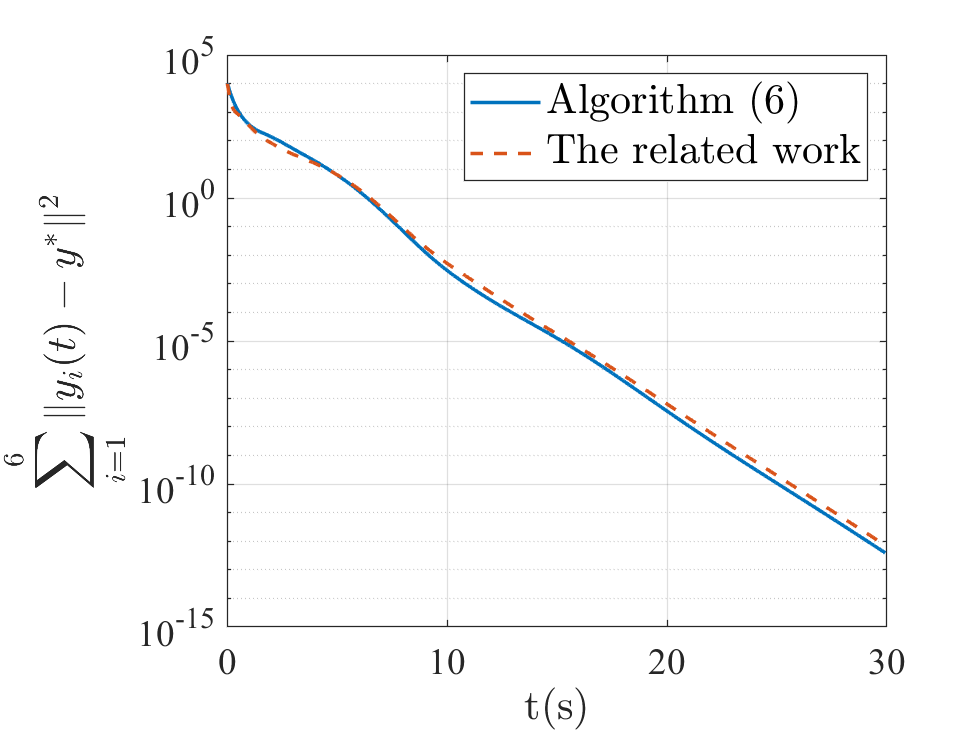} 
		\caption{The average outputs error.}
		\label{pic5}
	\end{figure}
	
	By taking the average of 20 runs, \autoref{pic5} shows that the two algorithms have a similar convergence, which both converge exponentially, while there is no analysis of convergence speed provided in \citet{li2020}.
\end{example}

\section{Results}\label{res}
This paper has investigated the optimal output consensus problem for heterogeneous linear multi-agent systems. A proportional-integral (PI) control law has been proposed, which can converge to the optimal solution exponentially. The proposed continuous algorithm does not require any global information, so it is fully distributed. Then, in order to avoid continuous communication among agents, the algorithm has been extended to periodic and event-triggered communication schemes. It was shown that the global exponential convergence is preserved and no Zeno behavior is exhibited.

Future works include extending the algorithms to the case of unbalanced directed and time-varying networks.

\bibliographystyle{elsarticle-harv}
\bibliography{bibfile}           

\begin{thebibliography}{20}
\expandafter\ifx\csname natexlab\endcsname\relax\def\natexlab#1{#1}\fi
\providecommand{\url}[1]{\texttt{#1}}
\providecommand{\href}[2]{#2}
\providecommand{\path}[1]{#1}
\providecommand{\DOIprefix}{doi:}
\providecommand{\ArXivprefix}{arXiv:}
\providecommand{\URLprefix}{URL: }
\providecommand{\Pubmedprefix}{pmid:}
\providecommand{\doi}[1]{\href{http://dx.doi.org/#1}{\path{#1}}}
\providecommand{\Pubmed}[1]{\href{pmid:#1}{\path{#1}}}
\providecommand{\bibinfo}[2]{#2}
\ifx\xfnm\relax \def\xfnm[#1]{\unskip,\space#1}\fi
\bibitem[{Ding et~al.({2018})Ding, Han, Ge and Zhang}]{ding2018}
\bibinfo{author}{Ding, L.}, \bibinfo{author}{Han, Q.L.}, \bibinfo{author}{Ge,
  X.}, \bibinfo{author}{Zhang, X.M.}, \bibinfo{year}{{2018}}.
\newblock \bibinfo{title}{An overview of recent advances in event-triggered
  consensus of multi-agent systems}.
\newblock \bibinfo{journal}{{IEEE Transactions on Cybernetics}}
  \bibinfo{volume}{{48}}, \bibinfo{pages}{{1110--1123}}.
\bibitem[{Ge et~al.(2020)Ge, Han, Ding, Wang and Zhang}]{ge2020}
\bibinfo{author}{Ge, X.}, \bibinfo{author}{Han, Q.L.}, \bibinfo{author}{Ding,
  L.}, \bibinfo{author}{Wang, Y.L.}, \bibinfo{author}{Zhang, X.M.},
  \bibinfo{year}{2020}.
\newblock \bibinfo{title}{Dynamic event-triggered distributed coordination
  control and its applications: A survey of trends and techniques}.
\newblock \bibinfo{journal}{IEEE Transactions on Systems, Man, and Cybernetics:
  Systems} \bibinfo{volume}{50}, \bibinfo{pages}{3112--3125}.
\bibitem[{Godsil and Royle(2001)}]{tu}
\bibinfo{author}{Godsil, C.}, \bibinfo{author}{Royle, G.F.},
  \bibinfo{year}{2001}.
\newblock \bibinfo{title}{Algebraic Graph Theory}.
\newblock \bibinfo{publisher}{New York, NY, USA: Springer}.
\bibitem[{Khalil and Grizzle(2002)}]{nonlinear}
\bibinfo{author}{Khalil, H.K.}, \bibinfo{author}{Grizzle, J.W.},
  \bibinfo{year}{2002}.
\newblock \bibinfo{title}{Nonlinear Systems}. volume~\bibinfo{volume}{3}.
\newblock \bibinfo{publisher}{Prentice hall Upper Saddle River, NJ}.
\bibitem[{Kia et~al.({2015})Kia, Cortes and Martinez}]{kia2015}
\bibinfo{author}{Kia, S.S.}, \bibinfo{author}{Cortes, J.},
  \bibinfo{author}{Martinez, S.}, \bibinfo{year}{{2015}}.
\newblock \bibinfo{title}{Distributed convex optimization via continuous-time
  coordination algorithms with discrete-time communication}.
\newblock \bibinfo{journal}{{Automatica}} \bibinfo{volume}{{55}},
  \bibinfo{pages}{{254--264}}.
\bibitem[{Li et~al.({2019}a)Li, Chen and Su}]{lxx2019}
\bibinfo{author}{Li, X.}, \bibinfo{author}{Chen, M.Z.Q.}, \bibinfo{author}{Su,
  H.}, \bibinfo{year}{{2019}}a.
\newblock \bibinfo{title}{Quantized consensus of multi-agent networks with
  sampled data and markovian interaction links}.
\newblock \bibinfo{journal}{{IEEE Transactions on Cybernetics}}
  \bibinfo{volume}{{49}}, \bibinfo{pages}{{1816--1825}}.
\bibitem[{Li et~al.({2019}b)Li, Xie and Hong}]{li20192}
\bibinfo{author}{Li, X.}, \bibinfo{author}{Xie, L.}, \bibinfo{author}{Hong,
  Y.}, \bibinfo{year}{{2019}}b.
\newblock \bibinfo{title}{Distributed continuous-time algorithm for a general
  nonsmooth monotropic optimization problem}.
\newblock \bibinfo{journal}{{International Journal of Robust and Nonlinear
  Control}} \bibinfo{volume}{{29}}, \bibinfo{pages}{{3252--3266}}.
\bibitem[{Li et~al.({2020}a)Li, Xie and Hong}]{li20191}
\bibinfo{author}{Li, X.}, \bibinfo{author}{Xie, L.}, \bibinfo{author}{Hong,
  Y.}, \bibinfo{year}{{2020}}a.
\newblock \bibinfo{title}{Distributed continuous-time nonsmooth convex
  optimization with coupled inequality constraints}.
\newblock \bibinfo{journal}{{IEEE Transactions on Control of Network Systems}}
  \bibinfo{volume}{{7}}, \bibinfo{pages}{{74--84}}.
\bibitem[{Li et~al.({2020}b)Li, Wu, Li and Ding}]{li2020}
\bibinfo{author}{Li, Z.}, \bibinfo{author}{Wu, Z.}, \bibinfo{author}{Li, Z.},
  \bibinfo{author}{Ding, Z.}, \bibinfo{year}{{2020}}b.
\newblock \bibinfo{title}{Distributed optimal coordination for heterogeneous
  linear multiagent systems with event-triggered mechanisms}.
\newblock \bibinfo{journal}{{IEEE Transactions on Automatic Control}}
  \bibinfo{volume}{{65}}, \bibinfo{pages}{{1763--1770}}.
\bibitem[{Nedic and Ozdaglar({2009})}]{nedic2009}
\bibinfo{author}{Nedic, A.}, \bibinfo{author}{Ozdaglar, A.},
  \bibinfo{year}{{2009}}.
\newblock \bibinfo{title}{Distributed subgradient methods for multi-agent
  optimization}.
\newblock \bibinfo{journal}{IEEE Transactions on Automatic Control}
  \bibinfo{volume}{{54}}, \bibinfo{pages}{{48--61}}.
\bibitem[{Qiu et~al.(2019)Qiu, Xie and You}]{qiu2019}
\bibinfo{author}{Qiu, Z.}, \bibinfo{author}{Xie, L.}, \bibinfo{author}{You,
  K.}, \bibinfo{year}{2019}.
\newblock \bibinfo{title}{Feedback-feedforward control approach to distributed
  optimization}, in: \bibinfo{booktitle}{American Control Conference (ACC)},
  pp. \bibinfo{pages}{1412--1417}.
\bibitem[{Tang et~al.({2019})Tang, Deng and Hong}]{tang2019}
\bibinfo{author}{Tang, Y.}, \bibinfo{author}{Deng, Z.}, \bibinfo{author}{Hong,
  Y.}, \bibinfo{year}{{2019}}.
\newblock \bibinfo{title}{Optimal output consensus of high-order multi-agent
  systems with embedded technique}.
\newblock \bibinfo{journal}{{IEEE Transactions on Cybernetics}}
  \bibinfo{volume}{{49}}, \bibinfo{pages}{{1768--1779}}.
\bibitem[{Tran et~al.(2019)Tran, Wang, Liu, Xiao and Lei}]{tran2018}
\bibinfo{author}{Tran, N.T.}, \bibinfo{author}{Wang, Y.W.},
  \bibinfo{author}{Liu, X.K.}, \bibinfo{author}{Xiao, J.W.},
  \bibinfo{author}{Lei, Y.}, \bibinfo{year}{2019}.
\newblock \bibinfo{title}{Distributed optimization problem for second-order
  multi-agent systems with event-triggered and time-triggered communication}.
\newblock \bibinfo{journal}{{Journal of the Franklin Institute}}
  \bibinfo{volume}{{356}}, \bibinfo{pages}{{10196--10215}}.
\bibitem[{Wu et~al.(2020)Wu, Li, Ding and Li}]{wu2020}
\bibinfo{author}{Wu, Z.}, \bibinfo{author}{Li, Z.}, \bibinfo{author}{Ding, Z.},
  \bibinfo{author}{Li, Z.}, \bibinfo{year}{2020}.
\newblock \bibinfo{title}{Distributed continuous-time optimization with
  scalable adaptive event-based mechanisms}.
\newblock \bibinfo{journal}{IEEE Transactions on Systems, Man, and Cybernetics:
  Systems} \bibinfo{volume}{50}, \bibinfo{pages}{3252--3257}.
\bibitem[{Yang et~al.({2019})Yang, Yi, Wu, Yuan, Wu, Meng, Hong, Wang, Lin and
  Johansson}]{yang2019}
\bibinfo{author}{Yang, T.}, \bibinfo{author}{Yi, X.}, \bibinfo{author}{Wu, J.},
  \bibinfo{author}{Yuan, Y.}, \bibinfo{author}{Wu, D.}, \bibinfo{author}{Meng,
  Z.}, \bibinfo{author}{Hong, Y.}, \bibinfo{author}{Wang, H.},
  \bibinfo{author}{Lin, Z.}, \bibinfo{author}{Johansson, K.H.},
  \bibinfo{year}{{2019}}.
\newblock \bibinfo{title}{A survey of distributed optimization}.
\newblock \bibinfo{journal}{Annual Reviews in Control} \bibinfo{volume}{{47}},
  \bibinfo{pages}{{278--305}}.
\bibitem[{Yu and Chen(2020)}]{yu2020}
\bibinfo{author}{Yu, H.}, \bibinfo{author}{Chen, T.}, \bibinfo{year}{2020}.
\newblock \bibinfo{title}{On {Z}eno behavior in event-triggered finite-time
  consensus of multi-agent systems}.
\newblock \bibinfo{journal}{IEEE Transactions on Automatic Control}
  \bibinfo{note}{Doi:{10.1109/TAC.2020.3030758}}.
\bibitem[{Zhang et~al.(2020)Zhang, Liu and Ji}]{zhang2020}
\bibinfo{author}{Zhang, J.}, \bibinfo{author}{Liu, L.}, \bibinfo{author}{Ji,
  H.}, \bibinfo{year}{2020}.
\newblock \bibinfo{title}{Exponential convergence of distributed optimal
  coordination for linear multi-agent systems over general digraphs}, in:
  \bibinfo{booktitle}{39th Chinese Control Conference (CCC)}, pp.
  \bibinfo{pages}{5047--5051}.
\bibitem[{Zhang et~al.(2018)Zhang, Papachristodoulou and Li}]{lina2017}
\bibinfo{author}{Zhang, X.}, \bibinfo{author}{Papachristodoulou, A.},
  \bibinfo{author}{Li, N.}, \bibinfo{year}{2018}.
\newblock \bibinfo{title}{Distributed control for reaching optimal steady state
  in network systems: An optimization approach}.
\newblock \bibinfo{journal}{IEEE Transactions on Automatic Control}
  \bibinfo{volume}{63}, \bibinfo{pages}{864--871}.
\bibitem[{Zhang et~al.({2017})Zhang, Deng and Hong}]{zhang2017}
\bibinfo{author}{Zhang, Y.}, \bibinfo{author}{Deng, Z.}, \bibinfo{author}{Hong,
  Y.}, \bibinfo{year}{{2017}}.
\newblock \bibinfo{title}{Distributed optimal coordination for multiple
  heterogeneous euler-lagrangian systems}.
\newblock \bibinfo{journal}{Automatica} \bibinfo{volume}{{79}},
  \bibinfo{pages}{{207--213}}.
\bibitem[{Zhao et~al.({2017})Zhao, Liu, Wen and Chen}]{zhao2017}
\bibinfo{author}{Zhao, Y.}, \bibinfo{author}{Liu, Y.}, \bibinfo{author}{Wen,
  G.}, \bibinfo{author}{Chen, G.}, \bibinfo{year}{{2017}}.
\newblock \bibinfo{title}{Distributed optimization for linear multi-agent
  systems: Edge- and node-based adaptive designs}.
\newblock \bibinfo{journal}{{IEEE Transactions on Automatic Control}}
  \bibinfo{volume}{{62}}, \bibinfo{pages}{{3602--3609}}.

\end{thebibliography}

\end{document}